\documentclass[preprint,a4paper,aip,jmp,author-year,tightenlines]{revtex4-1}
\usepackage[utf8]{inputenc}
\usepackage{amsmath}
\usepackage{amssymb}
\usepackage{amsthm}
\usepackage{mathtools} % \smashoperator
\usepackage{enumerate}

\PassOptionsToPackage{hyphens}{url} % word breaks in url
\usepackage[breaklinks=true]{hyperref} % vor harvard, Zeilenumbruch in ToC
\usepackage{natbib}
\usepackage{cleveref}
\swapnumbers % numbers in front
\newtheorem{theorem}{Theorem}
\newtheorem{lemma}[theorem]{Lemma}
\newtheorem{definition}[theorem]{Definition}
\newtheorem{note}[theorem]{Note} % with normal font
\let\oldnote\note
\renewcommand{\note}{\oldnote\normalfont}

\numberwithin{equation}{section}
\numberwithin{theorem}{section}

\newcommand{\wlim}{\operatorname{w-lim}}
\newcommand{\N}{\mathbb{N}}
\newcommand{\R}{\mathbb{R}}
\newcommand{\C}{\mathbb{C}}
\renewcommand{\H}{\mathcal{H}}
\newcommand{\vint}{v_\mathrm{int}}
\renewcommand{\d}{\,\mathrm{d}} % for integrals
\renewcommand{\i}{\mathrm{i}} % for imaginary unit
\newcommand{\e}{\mathrm{e}} % for exp
\DeclareRobustCommand{\onehalf}{\textstyle \frac{1}{2}}

\begin{document}
\title{Regularity for evolution equations with non-autonomous perturbations in Banach spaces}
\author{Markus Penz}
\email[corresponding author: ]{markus.penz@mpsd.mpg.de}
\affiliation{Max Planck Institute for the Structure and Dynamics of Matter, Hamburg, Germany}

\begin{abstract}
We provide regularity of solutions to a large class of evolution equations on Banach spaces where the generator is composed of a static principal part plus a non-autonomous perturbation. Regularity is examined with respect to the graph norm of the iterations of the principal part. The results are applied to the Schrödinger equation and conditions on a time-dependent scalar potential for regularity of the solution in higher Sobolev spaces are derived.
\\
(Published in a more compact form as: \textit{J.~Math.~Phys.}~\textbf{59}, 103512 (2018))
\end{abstract}
%\pacs{02.30.Jr, 02.30.Sa, 03.65.Db}
% PACS:
% 02 Mathematical methods in physics
% - Partial differential equations
% - Functional analysis
% 03 Quantum mechanics
% - Functional analytical methods

%\keywords{evolution equation, non-autonomous perturbation, Banach space, graph norm, regularity of solutions, Schrödinger equation, Sobolev space}

\maketitle
%\tableofcontents

\section{Introduction}

In this note we study the existence and regularity of solutions to the evolution equation
\begin{equation}\label{eq-evolution-equation}
\partial_t u(t) = G(t)u(t) \quad\mbox{with}\quad G(t) = A + B(t)
\end{equation}
for finite times $t\in [0,T]$ in a Banach space $X$. The generator $G(t)$ is composed of a principal static part $A$ and a potentially time-dependent (non-autonomous) perturbation $B(t)$. The conditions for existence of an evolution system (solution operator) $U(t,s)$ are formulated with respect to the domain $D(A)$ of the principal part. Since $A$ is assumed closed, this domain is a Banach space in itself if equipped with the graph norm $\|x\|_{D(A)} = \|x\|_X + \|Ax\|_X$ (\cref{lemma-closed-banach}). In fact we will show with \cref{th-regularity} that if the generator belongs to a (quasi)contraction semigroup and $B$ fulfils higher-order relative boundedness with respect to $A$ (\cref{def-rel-boundedness}) and a Lipschitz property \eqref{eq-lipschitz}, then $U(t,s)$ is a bounded operator $D(A^m) \rightarrow D(A^m)$. This stands in contrast to approaches that study regularity in pre-defined classes, whereas here the respective regularity class is directly yielded by the generator's principal part.

This result in the setting of an $N$-particle Hilbert space $X = L^2(\Omega^N)$ and Schrödinger-type equations leads to the regularity of propagated wave functions in terms of the graph norm of the iterated Laplacian. The perturbation $B(t)$ can be considered as consisting of external or inter-particle potentials of a certain Kato-perturbation type (\cref{def-sobolev-kato-perturbations}). Since the graph norm of $D(\Delta^m)$ is equivalent to the usual Sobolev space norm (\cref{th-sobolev-norm-laplace}), an important result considering the regularity of Schrödinger solutions with respect to Sobolev spaces is achieved (see final \cref{th-se}). The growth estimate for the Sobolev norm over time is of exponential type \eqref{eq-graph-norm-growth}. In the case of potentials that can be assumed smooth in space and time and periodic in space notable results concerning the (linear) growth of Sobolev norms have been achieved by \citet{bourgain-1999} and more recently by \citet{delort-2010}. A very nice recent work aimed at studying the time regularity of solutions to facilitate the Runge--Gross proof of time-dependent density functional theory \citep{tddft-review} can be found in \citet{fournais-2016}.

The existence part of the proof of our main \cref{th-regularity} is similar to the original treatment of \citet{kato-1953}, the usual reference point is \citet[Th.~X.70]{reed-simon-2}, but we give more general conditions that are also easier to check. This is in the spirit of a recent effort by \citet{schmid-griesemer-2014} to simplify and standardise the classical existence results for \eqref{eq-evolution-equation} to the simple condition that $t \mapsto G(t)x$ is continuously differentiable for all $x \in D(G)$. This was followed in \citet{schmid-griesemer-2016} by a generalisation to Lipschitz continuity which already bears some similarities to this work but includes no higher regularity estimates. Further \citet{schmid-griesemer-2016} built their argument around uniformly convex spaces which is, to our understanding, not necessary, since the more general notion of reflexivity is sufficient. Reflexivity is especially convenient because it is conserved when switching to equivalent norms, which is heavily used throughout this work, while uniform convexity is not.

Note that this e-print article includes more detailed explanations, additional proofs, and further references as compared to the published version.

\section{Graph-norm spaces}

The setting is always a Banach space $X$ with norm $\|x\| = \|x\|_X$. Domain and range (image) of an operator will be denoted as $D(A)$ and $R(A)$ respectively. We define the graph norm for a linear operator $A:X \rightarrow X$, generally unbounded, and $x \in D(A)$ as $\|x\|_{D(A)} = \|x\| + \|Ax\|$. This definition yields an equivalent norm to the more frequently given expression for the graph norm as a Pythagorean sum. The inequality
\begin{equation}
\sqrt{\|x\|^2 + \|Ax\|^2} \leq \|x\| + \|Ax\| \leq \sqrt{2}\sqrt{\|x\|^2 + \|Ax\|^2}
\end{equation}
is easily shown to hold by squaring it and using the inequality of arithmetic and geometric means, i.e., $\|x\| \, \|Ax\| \leq \onehalf(\|x\|^2 + \|Ax\|^2)$. The symbol ``$\sim$'' will later be used to denote equivalence of norms.

The domain $D(A^k)$ of the iterated operator is the set of all Banach space elements $x \in X$ where for all $1 \leq j \leq k$ also $A^j x \in X$ holds. (See \citet{fournais-2016} for a special emphasis on this in the context of studying the time-regularity of solutions to the Schrödinger equation with Coulomb potentials.) Note that $D(A^0) = D(I) = X$. Since the main tool of analysis will be higher-order graph-norm spaces $D(A^k)$, we adopt a shorthand notation for their norms and the respective operator norms.
\begin{align}
\label{def-graph-norm}
\|x\|_{(k)} &= \|x\|_{D(A^k)} = \|x\| + \|Ax\| + \ldots + \|A^kx\| \\
\|T\|_{(k,l)} &= \|T\|_{\mathcal{B}(D(A^k),D(A^{l}))} = \sup_{\substack{x \in D(A_k) \\ x \neq 0}} \frac{\|Tx\|_{(l)}}{\|x\|_{(k)}}
\end{align}
The parentheses in the subscript shall discern this notation from the usual $L^p$ and $W^{k,p}$ norms $\|\cdot\|_p$ and $\|\cdot\|_{k,p}$. We directly note the following chain of continuous embeddings.
\begin{equation}
D(A^k) \hookrightarrow D(A^{k-1}) \hookrightarrow \ldots \hookrightarrow D(A) \hookrightarrow X
\end{equation}
In correspondence with spaces equipped with the graph norm the notion of closed operators is all-important.

\begin{definition}\label{def-closed}
A linear and generally unbounded operator $A : X \rightarrow X$ is called \textbf{closed}, if for every sequence $x_n$ in $D(A)$ with $x_n \rightarrow x \in X$ and $Ax_n \rightarrow y \in X$ it holds that $x \in D(A)$ and $y=Ax$. An operator is called \textbf{closable} if it has a closed extension.
\end{definition}

\begin{note}\label{note-graph-closed}
An equivalent notion of closedness is often given in terms of the graph of the operator $\Gamma(A) = \{ (x,Ax) \mid x \in D(A) \} \subset X \times X$. $A$ is now closed if and only if its graph is closed as a subset of $X \times X$. The concept of graph norm is also derived from this picture, as the natural norm for elements of $\Gamma(A)$, and used in the following lemma for a further equivalence. \cite[Def.~B.1]{engel-nagel}
\end{note}

\begin{lemma}\label{lemma-closed-banach}
An operator $A$ is closed if and only if $D(A)$ equipped with the graph norm $\|\cdot\|_{D(A)}$ is a Banach space.
\end{lemma}

\begin{proof}
If $x_n \rightarrow x, Ax_n \rightarrow y$ are sequences as in \cref{def-closed} above then $x_n$ converges in graph norm and thus $x \in D(A), \|x_n-x\| \rightarrow 0$, and $\|Ax_n-Ax\| \rightarrow 0$. This establishes $y=Ax$ and thus closedness of $A$.\\
Starting from the definition of closedness we need to show completeness, i.e., every Cauchy sequence $x_n$ with respect to the graph norm converges in $D(A)$. This means $x_n$ and $Ax_n$ are Cauchy in $X$ and have the limits $x$ and $y$ respectively. Because $A$ is closed this yields $x \in D(A)$ as the proper limit of the sequence in $D(A)$.
\end{proof}

\begin{note}
In this sense a closed operator is always bounded as $A:D(A) \rightarrow X$ because clearly $\|Ax\| \leq \|x\|_{(1)}$ and \textit{closedness} can be seen as a notion of \textit{almost boundedness}. It is not even that simple to come up with examples of non-closable operators.\footnote{See StackExchange Mathematics: \url{https://math.stackexchange.com/questions/1811205} for an operator that even has dense graph, so the closure of the graph is $X \times X$ and really looks nothing like the graph of an operator.}
% and maybe: https://math.stackexchange.com/questions/819526/operator-not-closable
\end{note}

\begin{lemma}\label{lemma-inverse-closed}
A closed and injective operator has a closed inverse.
\end{lemma}

\begin{proof}
% (https://mathoverflow.net/questions/264740/non-empty-resolvent-set-then-operator-closed)
Since the restricted operator $A : D(A) \rightarrow R(A)$ is bijective we can define $A^{-1} : R(A) \rightarrow D(A)$. Now the graph $\Gamma(A) = \{ (x,Ax) \mid x \in D(A) \}$ can be rewritten with $y=Ax$ as $\Gamma(A) = \{ (A^{-1}y,y) \mid y \in R(A) \}$ which is isomorphic to $\Gamma(A^{-1})$. Thus from $\Gamma(A)$ closed follows $\Gamma(A^{-1})$ closed which is equivalent to closedness of operators as explained in \cref{note-graph-closed}.
\end{proof}

\begin{definition}\label{def-resolvent-set}
The \textbf{resolvent set} $\rho(A)$ of a closed operator $A$ is the set of all $\lambda \in \C$ such that $A-\lambda I$ has full range $X$ and a bounded inverse $(A-\lambda I)^{-1}$, called the \textbf{resolvent operator}. The elements of $\rho(A)$ are called \textbf{regular values} of $A$, their complement forms the \textbf{spectrum} of the operator.
\end{definition}

\begin{note}\label{note-closed-graph-th}
A popular alternative definition of the resolvent set merely demands a dense range for $A-\lambda I$ without mentioning closedness, see for example \citet[VIII.1]{yosida} or \citet[8.3]{renardy-rogers}. The resulting bounded resolvent operator is then densely defined and can in principle be continuously extended.
Since an operator that is everywhere defined is closed if and only if it is bounded (closed graph theorem), in our setting the resolvent operator is automatically closed. A closed, invertible operator has a closed inverse as shown in \cref{lemma-inverse-closed}, so one directly concludes $A-\lambda I$ and thus also $A$ closed from a non-empty resolvent set if \cref{def-resolvent-set} is used. Without the condition of full range in the case of a non-empty resolvent set the operator is still \textit{closable}. In any case it seems strange that closedness is listed here as an additional condition. This is because if one concentrates only on closed operators, like in our case and also in \citet[IV.1]{engel-nagel} or \citet[III.6.1]{kato-book}, then $R(A-\lambda I)=X$ is just an equivalent condition to a dense range, so the two possible definitions of resolvent sets actually agree again. To keep matters straight in what follows, both \textit{closedness} and \textit{non-empty resolvent set} will be noted as requirements side by side.\footnote{In these matters a discussion on StackExchange Mathematics has been proven very instructive: \url{http://math.stackexchange.com/questions/815377/resolvent-definition}}
\end{note}

\begin{lemma}\label{lemma-iterations-closed-1}
If $A$ is closed with non-empty resolvent set then all its iterations $A^k$, $k \in \mathbb{N}$, are closed.
\end{lemma}

\begin{proof}
Take $\lambda \in \rho(A)$ then the iterated resolvent operator $(A-\lambda I)^{-k}$ is everywhere defined and bounded and thus closed (closed graph theorem). Since a closed, invertible operator has a closed inverse (\cref{lemma-inverse-closed}) we also have $(A-\lambda I)^{k}$ closed. This expression expands to
\[
(A-\lambda I)^{k} = \sum_{j=0}^k {k \choose j} (-\lambda)^{k-j} A^{j}
\]
which facilitates an easy induction scheme. From $I, A, (A-\lambda I)^{2}$ closed we arrive at $A^2$ closed and so on.
\end{proof}

\section*{Graph-norm spaces from generators}

This unnumbered section is just an insert for the case of $A$ being the generator of a contraction semigroup. Since this implies that $A$ is closed and has a non-empty resolvent set, the results are in fact less general, but since we are often concerned with such operators they still bear some relevance. We refer to \cref{note-generator} for some comments on semigroups and their generators. Here we will first give an alternative proof for \cref{lemma-iterations-closed-1}, then show in \cref{lemma-equiv-norms} that the graph norm of $D(A^k)$ from \eqref{def-graph-norm} and its alternative definition $\|x\|+\|A^k x\|$ are equivalent.
%A referee's remark pointed out that this equivalence also holds in the case of a general $A$ with non-empty resolvent, but no such proof was found.
The result also holds for operators with zero in the resolvent set and it can be derived in a Hilbert space setting for $A$ self-adjoint as well. The main tool in the proof of \cref{lemma-equiv-norms} will be the weighted inequality of arithmetic and geometric means. The famous Stone theorem (\cref{th-stone}) explains that there is actually an intimate relation between contraction semigroups and self-adjoint operators. Conceptually between Hilbert and Banach spaces lie reflexive Banach spaces that are treated in the next section.

\begin{lemma}\label{lemma-iterations-closed-2}
If $A$ is the generator of a contraction semigroup then all its iterations $A^k$, $k \in \mathbb{N}$, are closed.
\end{lemma}

\begin{proof}
That the generator of a strongly continuous semigroup, which is even more general than assuming a contraction semigroup, is closed (and densely defined) is a classical result, see \citet[Th.~12.12]{renardy-rogers}, and will not be repeated here.\\
In a preparatory step we use the \citet{kallman-rota} inequality (a generalised form of the Landau--Kolmogorov inequality, see also \citet{hille-1972})
\begin{equation}\label{eq-kallman-rota}
\|A x\|^{2} \leq 4 \, \|x\|\, \|A^2 x\|
\end{equation}
that holds for any $x \in D(A^2)$ if $A$ is the generator of a \textit{contraction} semigroup. From this we want to prove the following relation for all $k \in \N$ and $x \in D(A^{k+1})$.
\begin{equation}\label{eq-kr-Ak}
\|A^k x\|^{(k+1)/k} \leq 2^{k+1} \|x\|^{1/k}\, \|A^{k+1} x\|
\end{equation}
The case $k=1$ is just \eqref{eq-kallman-rota} and we proceed by induction. We take \eqref{eq-kallman-rota} but replace $x$ by $A^k x$ to get
\begin{equation}
\|A^{k+1} x\|^{2} \leq 4 \, \| A^k x \| \, \|A^{k+2} x\|
\end{equation}
and in the next step use \eqref{eq-kr-Ak} that is assumed to hold to arrive at
\begin{equation}
\|A^{k+1} x\|^2 \leq 2^{k+2} \|x\|^{1/(k+1)}\, \|A^{k+1} x\|^{k/(k+1)} \|A^{k+2} x\|.
\end{equation}
Collecting the $A^{k+1}$ terms on one side we get
\begin{equation}
\|A^{k+1} x\|^{(k+2)/(k+1)} \leq 2^{k+2} \|x\|^{1/(k+1)} \|A^{k+2} x\|.
\end{equation}
which is just the desired result for the case $k+1$.\\
We use induction again to infer from $A, A^k$ closed that $A^{k+1}$ closed. Following \cref{def-closed} we take a sequence $\{x_n\}_n$ in $D(A^{k+1})$ with $x_n \rightarrow x$ and $A^{k+1}x_n \rightarrow y$. Now since both sequences converge they are also Cauchy sequences. If we take $A^k(x_n-x_m)$ then by the estimate \eqref{eq-kr-Ak} above 
\begin{equation}
\|A^k(x_n-x_m)\| \leq 2^{k} \|x_n-x_m\|^{1/(k+1)}\, \|A^{k+1} (x_n-x_m)\|^{k/(k+1)} \longrightarrow 0
\end{equation}
thus $\{A^k x_n\}_n$ is also Cauchy and converges to some $x' \in X$. But now we are able to invoke closedness of $A$ and state that since $A^k x_n \rightarrow x'$ and $A^{k+1}x_n = A(A^{k}x_n) \rightarrow y$ it must hold that $x' \in D(A)$ and $Ax'=y$. Because $A^k$ is closed as well we know from $x_n \rightarrow x$ and $A^k x_n \rightarrow x'$ that $x \in D(A^k)$ and $A^k x = x'$. A combination of these results yields just $x \in D(A^{k+1})$ and $A^{k+1}x=y$ and shows that $A^{k+1}$ is closed.
\end{proof}

\begin{lemma}\label{lemma-equiv-norms}
If $A$ is the generator of a contraction semigroup then the following equivalence of norms on $D(A^k)$ holds.
\begin{equation}
\|x\|_{(k)} \sim \|A^k x\| + \|x\|
\end{equation}
\end{lemma}

\begin{proof}
The proof uses the Kallman--Rota inequality again and is very similar to the proof of the previous lemma. That for $x \in D(A^k)$
\begin{equation}
\|x\|_{(k)} \geq \|A^k x\| + \|x\|
\end{equation}
is clear from the definition, so we only have to show
\begin{equation}\label{eq-equiv-norms-estimate}
\|x\|_{(k)} \leq C \left(\|A^k x\| + \|x\|\right)
\end{equation}
for some constant $C>0$. We already showed
\begin{equation}
\|A^{k-1} x\|^{k/(k-1)} = \|A^{k-1} x\|^{1/(k-1) + 1} \leq 2^{k} \|x\|^{1/(k-1)}\, \|A^{k} x\|
\end{equation}
in \eqref{eq-kr-Ak} where $k$ was just lowered by 1, which by the weighted inequality of arithmetic and geometric means \citep[(2.5.1)]{inequalities-book} yields
\begin{equation}
\|A^{k-1} x\| \leq 2^{k-1} \frac{\|x\| + (k-1)\|A^{k} x\|}{k}.
\end{equation}
A similar estimate can be derived analogously for all $\|A^{j} x\|$, $1 \leq j < k-1$, with only multiples of $\|x\|$ and $\|A^{k} x\|$ on the right hand side. This then already establishes the desired estimate \eqref{eq-equiv-norms-estimate}.
\end{proof}

\section{Reflexivity}

In this section reflexivity of Banach spaces is used to great advantage to get the properties of the limits of weakly converging sequences in \cref{lemma-weak-bounded} and \cref{lemma-weak-convergence} (which is actually a weaker version of lemma 5 in \citet{kato-1953} used there on p.~222). To be even able to work with those results later, we first need to establish reflexivity of graph-norm spaces.

\begin{lemma}\label{lemma-reflexive}
For $A:X \rightarrow X$ closed, $D(A)$ equipped with the graph norm is a reflexive Banach space if $X$ is so.\footnote{This was inspired by a short answer on StackExchange Mathematics: \url{http://math.stackexchange.com/questions/107721}.}
\end{lemma}

\begin{proof}
Take the isometry
\begin{equation}
i : D(A) \longrightarrow X \times X, x \longmapsto (x,Ax)
\end{equation}
where for $x \in D(A)$ we have the usual graph norm $\|x\| + \|Ax\|$ and for $(x_1,x_2) \in X\times X$ the sum norm $\|(x_1,x_2)\| = \|x_1\| + \|x_2\|$. Now $i(D(A))$ is closed in $X \times X$ as the image of a closed set under an isometry\footnote{See StackExchange Mathematics: \url{https://math.stackexchange.com/questions/1519704}.} and $X \times X$ is known to be reflexive if $X$ is. But any closed subspace of a reflexive space if reflexive itself, so $D(A)$ is.
\end{proof}

\begin{lemma}\label{lemma-weak-bounded}
Every weakly convergent sequence in $X$ is bounded.
\end{lemma}

\begin{proof}
If $\{ x_n \}_n \subset X$ is weakly convergent to $x \in X$ then it holds that for all $f \in X^*$ we have $f(x_n) \rightarrow f(x)$ so the sequence $\{ f(x_n) \}_n$ is clearly bounded in $\R$ or $\C$. This means we can find a respective $M_f > 0$ such that $\sup \{ |f(x_n)| \}_n \leq M_f$. But for every $x_n \in X$ we have a corresponding $z_n \in X^{**} \supseteq X$ such that $f(x_n) = z_n(f)$. Thus $\sup \{ |z_n(f)| \}_n \leq M_f$ for all $f \in X^*$ and finally because of the uniform boundedness principle (Banach--Steinhaus theorem) $\|z_n\| = \|x_n\| < \infty$ for all $n \in \N$.
\end{proof}

\begin{lemma}\label{lemma-weak-convergence}
Let $X$ be reflexive and $A:X \rightarrow X$ a closed operator with non-empty resolvent set. If a sequence $x_n \in D(A)$ has $\wlim_{n \rightarrow \infty} x_n = x \in X$ (weak limit) and $\{ \|Ax_n\| \}_n$ bounded then it holds $x \in D(A)$ and $\|Ax\| \leq \limsup_{n \rightarrow \infty} \|Ax_n\|.$
\end{lemma}

\begin{proof}
Choose $\lambda \in \rho(A)$ and $\tilde{A} = A-\lambda I$ then we have $\tilde{A}^{-1}$ bounded and the adjoint $(\tilde{A}^{-1})^*$ also exists as a bounded operator $X^* \rightarrow X^*$. The triangle inequality yields $\|\tilde{A}x_n\| \leq \|Ax_n\| + |\lambda| \|x_n\|$ so since $\{ \|Ax_n\| \}_n$ bounded and using \cref{lemma-weak-bounded} we have $\{ \|\tilde{A}x_n\| \}_n$ bounded (i.e., inside some closed ball that is compact in the weak topology iff $X$ reflexive, see \citet[V.4.2]{conway}). This means there must be a weakly convergent subsequence $\{ \tilde{x}_n \}_n$ with $y = \wlim_{n \rightarrow \infty} \tilde{A}\tilde{x}_n \in X$ (Eberlein--Šmulian theorem, a generalization of Bolzano--Weierstraß, see \citet[p.~141]{yosida}). We thus have with the dual pairing written as $(\cdot,\cdot)$ for any $f \in X^*$
\begin{equation}\label{lemma-reflexive-1}
\left( (\tilde{A}^{-1})^*f, \tilde{A}\tilde{x}_n \right) \longrightarrow \left( (\tilde{A}^{-1})^*f, y \right) = \left( f, \tilde{A}^{-1}y \right).
\end{equation}
On the other hand starting with the same expression we get
\begin{equation}\label{lemma-reflexive-2}
\left( (\tilde{A}^{-1})^*f, \tilde{A}\tilde{x}_n \right) = \left( f, \tilde{A}^{-1}\tilde{A}\tilde{x}_n \right) = \left( f, \tilde{x}_n \right) \longrightarrow \left( f, x \right).
\end{equation}
Now $(\tilde{A}^{-1})^* X^*$ is dense in $X^*$ because otherwise we can find a $0 \neq z \in X \simeq X^{**}$ (reflexive Banach space) such that for all $f \in X^*$
\begin{equation}
0 = \left( (\tilde{A}^{-1})^*f,z \right) = \left( f,\tilde{A}^{-1}z \right)
\end{equation}
which implies $\tilde{A}^{-1}z = 0$, thus $z=0$ which leads to a contradiction. With $(\tilde{A}^{-1})^* X^*$ dense in $X^*$ and the two expression \eqref{lemma-reflexive-1} and \eqref{lemma-reflexive-2} holding true for all $f \in X^*$ we identify $x = \tilde{A}^{-1}y$ which clearly means $x \in D(\tilde{A}) = D(A)$ and $\wlim_{n \rightarrow \infty} \tilde{A}\tilde{x}_n = \tilde{A}x$. Since $\wlim_{n \rightarrow \infty} x_n = x$ by assumption we also have $\wlim_{n \rightarrow \infty} A\tilde{x}_n = Ax$.\\
The dual space always contains an element $g \in X^*$ that yields $(g,Ax) = \|Ax\|$ and has dual norm 1 (Hahn--Banach theorem), so $|(g,A\tilde{x}_n)| \leq \|A\tilde{x}_n\|$. We know $(g,A\tilde{x}_n) \longrightarrow (g,Ax) = \|Ax\|$ and so finally $\|Ax\| \leq \limsup_{n \rightarrow \infty} \|A\tilde{x}_n\| \leq \limsup_{n \rightarrow \infty} \|Ax_n\|$.
\end{proof}

%
%% already mentioned with Kakutani (weakly compact ball)
%It probably relates strongly to the statement: ``The closed unit ball in a normed space X is compact in the weak topology if and only if X is reflexive.'' (\href{https://en.wikipedia.org/wiki/Weak_topology\#Other_properties_2}{Wikipedia})
%
%A similar statement to the lemma can be found on StackExchange Mathematics: \url{http://math.stackexchange.com/questions/85023}.

\section{Relative boundedness and a preparatory lemma}

A final preparatory lemma will help us to establish the desired regularity result in \cref{sec-regularity} and gives the main condition on the non-autonomous perturbation $B(t)$ in the form of higher-order relative boundedness with respect to $A$ and a Lipschitz property. The assumption that all $A^k$ are closed is imperative to have Banach spaces $D(A^k)$ from \cref{lemma-closed-banach}. In our main \cref{th-regularity} we will then not have this condition explicitly, but rely on \cref{lemma-iterations-closed-1} to get closedness. We now define the new notion of \textit{$A$-boundedness of order $m$}, a generalization of the usual \textit{$A$-boundedness} (\citet[X.2]{reed-simon-2}, \citet[IV.1.1]{kato-book}) that corresponds to order $k=1$.

\begin{definition}\label{def-rel-boundedness}
Let $A,B$  be densely defined operators and $A$ closed, then $B$ is called \textbf{$A$-bounded of order $m$} if
\begin{enumerate}[(i)]
	\itemsep0em
	\item $D(A) \subseteq D(B)$ and
	\item there are $a,b \geq 0$ such that for all $1 \leq k \leq m$ and $x \in D(A^k)$ it holds that $\|B x\|_{(k-1)} \leq a \|A^k x\| + b\|x\|_{(k-1)}$.
\end{enumerate}
The infimum of possible values for $a$ is called the \textbf{relative bound} (of order $m$) of $B$ with respect to $A$. This relative bound can be as low as 0, the operator is then called \textbf{infinitesimally small}.
\end{definition}

\begin{note}\label{note-B-bounded}
From this definition it directly follows that if $B$ is $A$-bounded of order $m$, then for all $0 \leq k \leq m$ we have $B : D(A^k) \rightarrow D(A^{k-1})$ bounded with respect to the $\|\cdot\|_{(k,k-1)}$ operator norm. Thus this property does not need to be separately demanded in the following lemma. Note further that the usual $A$-boundedness of $B$ with relative bound strictly smaller than $1$ also means that $A+B$ is closed if and only if $A$ is. \citet[Th.~IV.1.1]{kato-book} Another very similar application of relative boundedness is the famous Kato--Rellich theorem that enters \cref{th-kato} here. A further consequence expressed in the following lemma is that the spaces $D(A^k)$ and $D(G^k)$, where $G=A+B$, become equivalent.
\end{note}

\begin{lemma}\label{lemma-AG-equiv}
Assume $A$ closed with non-empty resolvent set and $B$ to be $A$-bounded of order $m$ with relative bound $<1$. Then for $G=A+B$ it holds $D(G^k) = D(A^k)$ and the respective graph norms are equivalent for all $1\leq k \leq m$.
\end{lemma}

\begin{proof}
We start with $k=1$ and assume $x\in D(A)$, then $A$-boundedness yields
\begin{equation}\label{eq-AG-equiv-1}
\|Gx\| \leq \|Ax\| + \|Bx\| \leq (1+a)\|Ax\| + b\|x\|
\end{equation}
and thus $x\in D(G)$. On the other hand assume $x\in D(G)$, then one has from $A=G-B$
\begin{equation}
\|Ax\| \leq \|Gx\| + \|Bx\| \leq \|Gx\| + a\|Ax\| + b\|x\|
\end{equation}
and because of $a<1$ by assumption
\begin{equation}\label{eq-AG-equiv-2}
\|Ax\| \leq \frac{1}{1-a}\|Gx\| + \frac{b}{1-a}\|x\|
\end{equation}
and thus $x\in D(A)$. This establishes $D(G)=D(A)$ and further by the inequalities \eqref{eq-AG-equiv-1} and \eqref{eq-AG-equiv-2} equivalence of the norms.\\
Next we proceed by induction and assume the argument holds for $k-1$. Let $x\in D(A^k)$ then
\begin{equation}
\begin{aligned}\label{eq-AG-equiv-3}
\|Gx\|_{(k-1)} &\leq \|Ax\|_{(k-1)} + \|Bx\|_{(k-1)} \leq \|Ax\|_{(k-1)} + a\|A^k x\| + b\|x\|_{(k-1)} \\
&\leq (1+a+b) \|x\|_{(k)}.
\end{aligned}
\end{equation}
Since we already established $D(G^{k-1})=D(A^{k-1})$ with equivalent norms the expression $\|Gx\|_{(k-1)} + \|x\|$ corresponds to the graph norm of $D(G^{k})$ (up to a multiplicative constant). On the other hand let $x\in D(G^k)$ then
\begin{equation}
\begin{aligned}
\|A^k x\| + \|x\|_{(k-1)} - \|x\| = \|Ax\|_{(k-1)} &\leq \|Gx\|_{(k-1)} + \|Bx\|_{(k-1)} \\
&\leq \|Gx\|_{(k-1)} + a\|A^k x\| + b\|x\|_{(k-1)}
\end{aligned}
\end{equation}
and after a simple manipulation
\begin{equation}\label{eq-AG-equiv-4}
(1-a) \|x\|_{(k)} \leq (1-a)\|A^k x\| + \|x\|_{(k-1)} \leq \|Gx\|_{(k-1)} + \|x\| + b\|x\|_{(k-1)}.
\end{equation}
The right hand side corresponds again to the graph norm of $D(G^{k})$ and thus by inequalities $\eqref{eq-AG-equiv-3}$ and $\eqref{eq-AG-equiv-4}$ we get $D(G^{k})=D(A^{k})$ and equivalence of the respective norms.
\end{proof}

\begin{lemma}\label{lemma-K-bounded}
Let $G(t)=A+B(t)$ be closed with $0 \in \rho(G(t))$ at all times under consideration. Further assume $A^k$ to be closed for all $1\leq k\leq m$, $B(t)$ to always be $A$-bounded of order $m$ with a maximal relative bound strictly smaller than $1$, and demand the Lipschitz condition
% (this was previously uniform boundedness in time, but this is propably not needed at all)
\begin{equation}
L_m := \sup_{t'\neq t} \sum_{k=1}^m \frac{\|B(t')-B(t)\|_{(k,k-1)}}{|t'-t|} < \infty.
\end{equation}\label{eq-lipschitz}
Then $G(t')^m G(t)^{-m} = K_m(t',t) + I$ with the operator $K_m(t',t) : X \rightarrow X$ bounded by $C_m L_m |t'-t|$ if $t'\neq t$.
(See \eqref{eq-def-Cm} for a definition of the $G$-dependent constants $C_m$.)
\end{lemma}

\begin{proof}
In the whole proof we mostly write $G=G(t), G'=G(t')$ and similarly for $B$ for brevity. For all $1 \leq k \leq m$ clearly $A : D(A^k) \rightarrow D(A^{k-1})$ bounded by definition and $B$ bounded equally by assumption (\cref{note-B-bounded}). This makes the combined operator $G: D(A^k) \rightarrow D(A^{k-1})$ bounded as well. Because of zero in the resolvent set we have a well-defined and bounded $G^{-1} : X \rightarrow X$. We continue with the restriction of $G^{-1}$ on $D(A^{k-1})$ that will still be denoted $G^{-1}$. Because $B$ is $A$-bounded of order $m$ we get for all $x \in D(A^{k-1}) =D(G^{k-1})$, $G^{-1}x \in D(G^{k}) = D(A^{k})$ ({\cref{lemma-AG-equiv}})
\begin{equation}
\begin{aligned}
\|A^k G^{-1}x\| &= \|A^{k-1} (G-B) G^{-1}x\| \\
&\leq \|B G^{-1}x\|_{(k-1)} + \|x\|_{(k-1)}\\
&\leq a \|A^k G^{-1}x\| + b\|G^{-1}x\|_{(k-1)} + \|x\|_{(k-1)}
\end{aligned}
\end{equation}
and since $a<1$ it holds
\begin{equation}
\|A^k G^{-1}x\| \leq \frac{b}{1-a} \|G^{-1}x\|_{(k-1)} + \frac{1}{1-a} \|x\|_{(k-1)}.
\end{equation}
This result establishes an estimate
\begin{equation}
\begin{aligned}
\|G^{-1}x\|_{(k)} &= \|G^{-1}x\|_{(k-1)} + \|A^k G^{-1}x\| \\
&\leq \left(\frac{b}{1-a}+1\right) \|G^{-1}x\|_{(k-1)} + \frac{1}{1-a} \|x\|_{(k-1)}
\end{aligned}
\end{equation}
and after $k-1$ further iterations yields $G^{-1} : D(A^{k-1}) \rightarrow D(A^k)$ bounded. 
% bounded by the bounded inverse theorem \citep[8.2]{renardy-rogers} **not needed any more**
The iterated operators $G^k : D(A^k) \rightarrow X$, $G^{-k} : X \rightarrow D(A^k)$ are thus bounded as well, just like $K_k(t',t) : X \rightarrow X$ defined by $K_k(t',t) := G(t')^k G(t)^{-k} - I$. We still have to show the special bound for $K_m$.
We start with $k=1$.
\begin{equation}
K_1 = G'G^{-1}-I = (G'-G)G^{-1} = (B'-B)G^{-1}
\end{equation}
The operator $K_1$ has the bound
\begin{equation}
\|K_1\|_{(0,0)} = \|(B'-B)G^{-1}\|_{(0,0)} \leq \|B'-B\|_{(1,0)} \|G^{-1}\|_{(0,1)}.
\end{equation}
Now from $k=1$ on we proceed by iteration.
\begin{equation}
K_k = G'^kG^{-k}-I = G'^{k-1}(G'-G)G^{-k} + K_{k-1} = G'^{k-1}(B'-B)G^{-k} + K_{k-1}
\end{equation}
The estimate is now
\begin{equation}
\begin{aligned}
\|K_k\|_{(0,0)} &\leq \|G'^{k-1}(B'-B)G^{-k}\|_{(0,0)} + \|K_{k-1}\|_{(0,0)} \\
&\leq \|G'^{k-1}\|_{(k-1,0)} \|B'-B\|_{(k,k-1)} \|G^{-k}\|_{(0,k)} + \|K_{k-1}\|_{(0,0)}
\end{aligned}
\end{equation}
which sums up to
\begin{equation}\label{eq-def-Cm}
\begin{aligned}
\|K_m\|_{(0,0)} &\leq \sum_{k=1}^m \|G'^{k-1}\|_{(k-1,0)} \|B'-B\|_{(k,k-1)} \|G^{-k}\|_{(0,k)} \\
&\leq \max_{1 \leq k \leq m} \left\{ \|G'^{k-1}\|_{(k-1,0)} \|G^{-k}\|_{(0,k)} \right\} \times \sum_{k=1}^m \frac{\|B'-B\|_{(k,k-1)}}{|t'-t|} \times |t'-t| \\
&\leq \underbrace{\sup_{t'\neq t}\max_{1 \leq k \leq m} \left\{ \|G'^{k-1}\|_{(k-1,0)} \|G^{-k}\|_{(0,k)} \right\}}_{C_m} \times \underbrace{\sup_{t'\neq t}\sum_{k=1}^m \frac{\|B'-B\|_{(k,k-1)}}{|t'-t|}}_{L_m} \times |t'-t|.
\end{aligned}
\end{equation}
\end{proof}

%\begin{note}\label{note-spectrum-shift}
%The $0 \in \rho(G(t))$ condition is some real $\alpha \in \rho(G(t))$ without loss of generality, but this still has to hold for all $t$. A unitary but time-dependent transformation to shift a time-dependent $\alpha(t)$ to 0 is also possible with
%\begin{equation}
%u \longmapsto \tilde{u} = \exp\left( -\int_0^t \alpha(s) ds \right) u.
%\end{equation}
%Then $u$ solving $\partial_t u = G u$ is equivalent to $\tilde{u}$ solving $\partial_t \tilde{u} = (G-\alpha)\tilde{u}$.
%So it seems the condition for $G(t)$ is that it is bounded below or just has a energy gap (self-adjoint case for $iG(t)$) that includes an (not even continuous) $\alpha(t)$ at all times $t$. This shifts the energy of the system with $\alpha(t)$ and could lead to solutions that are not time-differentiable of order $\geq 2$.
%\end{note}

\section{Regularity result}
\label{sec-regularity}

First we will stick to a pure Banach space setting even though the strongest motivation is of course the Schrödinger equation that will be discussed in \cref{sec-se}. The domain of the time-dependent generator $G(t)$ is usually assumed to remain constant in time and we write $D(G(t)) = D(G)$. This is true anyway by assumption following \cref{lemma-K-bounded} where $D(G(t))=D(A)$ always holds.

%\begin{note}
%Give results for \textit{m-dissipative} operators and perturbations. This also links to Stone's theorem and Kato-Rellich.
%\end{note}

\begin{definition}\label{def-evolution-system}
An \textbf{evolution system} belonging to an evolution equation like \eqref{eq-evolution-equation} is a two-parameter family of bounded linear operators $U(t,s)$, $0 \leq s \leq t \leq T$, on $X$ that fulfils \citep[ch.~5, Def.~5.3]{pazy}
\begin{enumerate}[(i)]
  \itemsep0em
  \item $U(s,s) = I, U(t,r)U(r,s)=U(t,s)$ for all $0 \leq s \leq r \leq t \leq T$,
  \item $(t,s) \mapsto U(t,s)$ is jointly strongly continuous, i.e., $\lim_{t\rightarrow s}U(t,s)=I$ strongly and equivalently for $s \rightarrow t$,
\end{enumerate}
and on $D(G)$ solves
\begin{equation}
\begin{aligned}
  \partial_t U(t,s) &= G(t)U(t,s) \\
  \partial_s U(t,s) &= -U(t,s)G(s).
\end{aligned}
\end{equation}
\end{definition}

\begin{note}\label{note-unitary-evolution-system}
In the Hilbert space setting with $H(t)=\i G(t)$ self-adjoint (i.e., $G(t)$ is \textit{skew-adjoint}) that will be discussed in \cref{sec-se} the unitarity condition $U(t,s)^* = U(t,s)^{-1} = U(s,t)$ is added to the properties of the evolution system.
\end{note}

\begin{note}\label{note-generator}
In the case of a time-independent generator $G$ the typical setting is that of a strongly continuous one-parameter semigroup \citep{renardy-rogers,engel-nagel,pazy}. The generator of a strongly continuous semigroup is always densely defined and closed (\citet[Th.~12.12]{renardy-rogers}; this was already noted in the proof of \cref{lemma-iterations-closed-2}). A quasicontraction semigroup $U(t)$ is a strongly continuous semigroup with a constant $\omega > 0$ such that for all $x \in X$ and $t \geq 0$
\begin{equation}
\|U(t)x\| \leq \e^{\omega t}\|x\|.
\end{equation}
Clearly in the Hilbert space setting with a unitary (semi-)group it always holds $\|U(t)x\| \leq \|x\|$ which makes $U(t)$ trivially a contraction semigroup with $\omega=0$. The generator of a quasicontraction semigroup is known to have all $\lambda \in \C$ with $\mathrm{Re}\,\lambda > \omega$ in the resolvent set as a corollary to the famous Hille--Yosida generation theorem \citep[Cor.~II.3.6]{engel-nagel}. Thus $G$ as such a generator just needs to be shifted by $\omega+1$ to have $0 \in \rho(G)$ as demanded by \cref{lemma-K-bounded}. What follows is our main result regarding the regularity property of evolution systems.
\end{note}

\begin{theorem}\label{th-regularity}
On $X$ a reflexive Banach space let $G(t)=A+B(t)$ be the generator of a quasicontraction semigroup for all $t \in [0,T]$ with common contraction constant $\omega$, let $A$ be closed with non-empty resolvent set, and let $B$ have the properties of \cref{lemma-K-bounded}.
Then the evolution equation \eqref{eq-evolution-equation} has a well-defined evolution system given by the limit of the stepwise static approximation \eqref{eq-stepwise-approx} that is bounded $D(A^k) \rightarrow D(A^k)$ for all $0 \leq k \leq m$ and thus preserves regularity of Banach space vectors in the class $D(A^m)$.
\end{theorem}

\begin{proof}
First note that because of $A$ closed with non-empty resolvent set we have by \cref{lemma-iterations-closed-1} all $A^k$ closed and thus can work in Banach spaces $D(A^k)$.\\
Let $P_k$ be a sequence of equidistant partitions of $[0,T]$ with $k$ subintervals $[t_i,t_{i+1}]$ with $t_0=0$, $t_k=T$, and mesh size $T/k$ going to zero as $k \rightarrow \infty$. We write $\lfloor s \rfloor_k$ for the largest $t_i$ in the partition $P_k$ smaller or equal than $s \in [0,T]$. We define the stepwise static approximation to the evolution system by combining the $k$ individual evolution semigroups ${U_k}^{(i)}$ defined by the static generators $G(t_{i})$, $0 \leq i \leq k-1$.
\begin{equation}\label{eq-stepwise-approx}
\begin{aligned}
U_k(t,s) &= U_k^{(i)}(t-s) \quad \mathrm{if} \; t_{i} \leq s \leq t \leq t_{i+1} \\
U_k(t,s) &= U_k(t,t_i)U_k(t_i,s) \quad \mathrm{with} \; s < t_i < t \; \mathrm{else.}
\end{aligned}
\end{equation}
We show now that for $k\rightarrow \infty$ the $U_k(t,s)$ converges uniformly in $t$ and preserves the desired degree of regularity. Convergence is tested with the Cauchy property of the sequence $U_k$. We use $\partial_t U_k(t,s) = G(\lfloor t \rfloor_k)U_k(t,s)$ and $\partial_s U_k(t,s) = -U_k(t,s)G(\lfloor s \rfloor_k)$ which follows directly from the definition of $U_k$ above and the evolution semigroup property of ${U_k}^{(i)}$.
\begin{equation}\label{eq-cauchy-sequence}
\begin{aligned}
U_k(t,s)-U_l(t,s) &= U_l(t,r) U_k(r,s)\Big|_{r=s}^t \\
&= \int_s^t \partial_r (U_l(t,r) U_k(r,s)) \d r \\
&= \int_s^t \big((\partial_r U_l(t,r)) U_k(r,s) + U_l(t,r) (\partial_r U_k(r,s))\big) \d r \\
&= - \int_s^t U_l(t,r) \big(G(\lfloor r \rfloor_l) - G(\lfloor r \rfloor_k)\big) U_k(r,s) \d r \\
&= - \int_s^t U_l(t,r) \big(B(\lfloor r \rfloor_l) - B(\lfloor r \rfloor_k)\big) U_k(r,s) \d r
\end{aligned}
\end{equation}
(Note that a problem arises with the time derivative if $r=t_i$ because then the right and left derivatives do not match. But this is just at a finite number of points that can always be omitted from the integral.)
We check the Cauchy property in the $\|\cdot\|_{(m,m-1)}$ norm.
\begin{equation}\label{eq-cauchy-sequence-norm}
\|U_k(t,s)-U_l(t,s)\|_{(m,m-1)} \leq \int_s^t \|U_l(t,r)\|_{(m-1,m-1)} \|B(\lfloor r \rfloor_l) - B(\lfloor r \rfloor_k)\|_{(m,m-1)} \|U_k(r,s)\|_{(m,m)} \d r
\end{equation}
Since $B(t)$ has the Lipschitz-property up to order $m$ which implies continuity in the $\mathcal{B}(D(A^m),D(A^{m-1}))$ norm, the difference would go to zero if $k \rightarrow \infty$. But are $\|U_l(t,r)\|_{(m-1,m-1)}$ and $\|U_k(r,s)\|_{(m,m)}$ uniformly bounded? We will test for $\|U_k(r,s)\|_{(m,m)} < \infty$, all lower orders at all times $0 \geq s \geq r \geq T$ apply equally.\\
The idea is to switch to a shifted, auxiliary generator $\tilde{G}(t) = G(t) - (\omega +1)$ such that $0 \in \rho(\tilde{G}(t))$ for all $t$ (also see \cref{note-generator}) thus achieving accordance with the conditions of \cref{lemma-K-bounded}. Then we introduce the identities $\tilde{G}^{-m}(t_i)\tilde{G}^m(t_i)$ in front of all the short-time evolution operators, exchange them with the evolution operators originating from the same generator, and give an estimate for the arising terms involving $K_m$ by \cref{lemma-K-bounded} applied to $\tilde{G}$. Let $i$ be such that $\lfloor s \rfloor_k = t_i$, i.e., $i=\lfloor sk/T \rfloor$ (usual floor function brackets), and $j$ be such that $\lfloor r \rfloor_k = t_j$, i.e., $j=\lfloor rk/T \rfloor$.
\begin{equation}
\begin{aligned}
U_k(r,s) =& U_k^{(j)}(r-t_j) U_k^{(j-1)}(t_j-t_{j-1}) \ldots U_k^{(i)}(t_{i+1}-s) \\
=& \tilde{G}^{-m}(t_j)U_k^{(j)}(r-t_j)\tilde{G}^m(t_j) \tilde{G}^{-m}(t_{j-1})U_k^{(j-1)}(t_j-t_{j-1}) \tilde{G}^{m}(t_{j-1}) \ldots \\
&\ldots \tilde{G}^{-m}(t_{i+1})U_k^{(i+1)}(t_{i+2}-t_{i+1}) \tilde{G}^{m}(t_{i+1}) \tilde{G}^{-m}(t_i)U_k^{(i)}(t_{i+1}-s) \tilde{G}^{m}(t_i)
\end{aligned}
\end{equation}
We write the $\tilde{G}(t')^m \tilde{G}(t)^{-m}$ encounters like in \cref{lemma-K-bounded} as $K_m(t',t)+I$.
\begin{equation}
\begin{aligned}
U_k(r,s) = \tilde{G}^{-m}(t_j) &U_k^{(j)}(r-t_j) \prod_{l=i+1}^{j-1} \Big( K_m(t_{l+1},t_{l})+I \Big) U_k^{(l)}(t_{l+1}-t_l) \\
&\Big( K_m(t_{i+1},t_{i})+I \Big) U_k^{(i)}(t_{i+1}-s) \tilde{G}^{m}(t_i)
\end{aligned}
\end{equation}
(Note that the product is time-ordered.) Next we estimate the $\|\cdot\|_{(m,m)}$ norm of this expression using the result from \cref{lemma-K-bounded} and regular mesh size $t_i-t_{i-1} = T/k$. For this we repeatedly use the quasicontraction property that assures $\|U^{(l)}_k(t_{l+1}-t_l)\|_{(0,0)} \leq \e^{\omega (t_{l+1}-t_l)}$.
\begin{equation}\label{eq-mm-bound}
\begin{aligned}
\|U_k(r,s)\|_{(m,m)} &\leq \e^{\omega(r-s)} \|\tilde{G}^{-m}(t_j)\|_{(0,m)} \prod_{l=i}^{j-1} \Big( \|K_m(t_{l+1},t_{l})\|_{(0,0)} + 1 \Big) \|\tilde{G}^{m}(t_i)\|_{(m,0)} \\
&\leq \e^{\omega(r-s)} \|\tilde{G}^{-m}(t_j)\|_{(0,m)} \|\tilde{G}^{m}(t_i)\|_{(m,0)} \left( \frac{C_m L_m T}{k} + 1 \right)^{j-i}
\end{aligned}
\end{equation}
We rewrite $j-i = \lfloor rk/T \rfloor - \lfloor sk/T \rfloor \leq 1+ \lfloor (r-s)k/T \rfloor = 1+ k \lfloor (r-s)k/T \rfloor/k$ to be able to introduce an exponential function in the limit $k \rightarrow \infty$ while $\lfloor (r-s)k/T \rfloor/k \rightarrow (r-s)/T$.
\begin{equation}\label{eq-mm-bound-2}
\begin{aligned}
\left( \frac{C_m L_m T}{k} + 1 \right)^{j-i} &\leq \left( \frac{C_m L_m T}{k} + 1 \right) \left(\left( \frac{C_m L_m T}{k} + 1 \right)^k \right)^{\lfloor (r-s)k/T \rfloor/k} \\[0.5em]
&\longrightarrow e^{C_m L_m (r-s)} \leq e^{C_m L_m T}
\end{aligned}
\end{equation}
This means that \eqref{eq-cauchy-sequence-norm} goes to zero and the Cauchy sequence \eqref{eq-cauchy-sequence} must converge to a well-defined and bounded $U(t,s) : D(A^m) \rightarrow D(A^{m-1})$. In the lowest order $m=1$ the operators $U(t,s) : D(A) \rightarrow X$ can then be continuously extended from the dense $D(A)$ to the whole space $X$ because they are clearly bounded on $X$ (as a combination of quasicontraction semigroup elements). \Cref{lemma-weak-convergence} helps us to establish $U(t,s):D(A^m) \rightarrow D(A^{m})$ bounded as the desired regularity result.\footnote{This result is actually missing in the proof in \citet[Th.~3.41, Th.~3.42]{penz-phd}, where the boundedness was just directly inferred from the $U_k$ boundedness.} For this we take $x = U(t,s)\varphi$ and $x_k = U_k(t,s)\varphi$ for any $\varphi \in D(A^m)$ and the Banach space $D(A^{m-1})$ (denoted $X$ in the lemma) which is a reflexive Banach space by \cref{lemma-reflexive}. The weak limit clearly follows as a result of even strong convergence\footnote{So there is room for a more general statement, since at this spot only weak convergence would be sufficient.} in $D(A^{m-1})$ and $\{ \|Ax_k\|_{(m-1)} \}_k$ bounded follows from $\{\|x_k\|_{(m)}\}_k$ bounded which was just shown in \eqref{eq-mm-bound} above. Then the assertion of \cref{lemma-weak-convergence} says
\begin{equation}\label{eq-limsup-estimate}
\|Ax\|_{(m-1)} \leq \limsup_{k \rightarrow \infty} \|Ax_k\|_{(m-1)}.
\end{equation}
We first deal with the case $m=1$ which means for the inequality above
\begin{equation}
\|AU(t,s)\varphi\| \leq \limsup_{k \rightarrow \infty} \|AU_k(t,s)\varphi\|
\end{equation}
and thus by introducing the $\|\cdot\|_{(1)}$ norm
\begin{equation}
\begin{aligned}
\|U(t,s)\varphi\|_{(1)} &= \|U(t,s)\varphi\| + \|AU(t,s)\varphi\| \leq \|U(t,s)\varphi\| + \limsup_{k \rightarrow \infty} \|AU_k(t,s)\varphi\| \\
&= \limsup_{k \rightarrow \infty} \|U_k(t,s)\varphi\|_{(1)} + \|U(t,s)\varphi\| - \limsup_{k \rightarrow \infty} \|U_k(t,s)\varphi\|.
\end{aligned}
\end{equation}
In the limit the last two terms cancel because we already showed that $U_k(t,s) \rightarrow U(t,s)$ converges on $X$. So using the estimates from \eqref{eq-mm-bound} and \eqref{eq-mm-bound-2} we get
\begin{equation}
\|U(t,s)\|_{(1,1)} \leq C'_1 \exp(C_1 L_1 T)
\end{equation}
where the additional constants from \eqref{eq-mm-bound} have been collected in $C'_1$. This means $U_k(t,s) \rightarrow U(t,s)$ converges also as a bounded operator $D(A) \rightarrow D(A)$. The next step is already for arbitrary $m$ and we use \eqref{eq-limsup-estimate} again.
\begin{equation}
\begin{aligned}
\|U(t,s)\varphi\|_{(m)} &= \|U(t,s)\varphi\| + \|A^m U(t,s)\varphi\| 
= \|U(t,s)\varphi\| + \|A^{m-1} A U(t,s)\varphi\| \\
&= \|U(t,s)\varphi\| + \|A U(t,s)\varphi\|_{(m-1)} - \|A U(t,s)\varphi\| \\
&\leq \|U(t,s)\varphi\| + \limsup_{k \rightarrow \infty} \|AU_k(t,s)\varphi\|_{(m-1)} - \|A U(t,s)\varphi\|
\end{aligned}
\end{equation}
Next the $D(A^{m-1})$ norm gets rewritten to a $D(A^m)$ norm.
\begin{equation}
\begin{aligned}
\|U(t,s)\varphi\|_{(m)} &\leq \limsup_{k \rightarrow \infty} \|U_k(t,s)\varphi\|_{(m)} \\
&+ \|U(t,s)\varphi\| - \limsup_{k \rightarrow \infty}\|U_k(t,s)\varphi\| \\
&- \|A U(t,s)\varphi\| + \limsup_{k \rightarrow \infty} \|A U_k(t,s)\varphi\|
\end{aligned}
\end{equation}
This time the whole two last lines vanish in the limit because in the meantime we also established convergence on $D(A)$, so with the estimates from \eqref{eq-mm-bound} and \eqref{eq-mm-bound-2} we finally get boundedness $D(A^m) \rightarrow D(A^m)$.
\begin{equation}\label{eq-graph-norm-growth}
\|U(t,s)\|_{(m,m)} \leq C'_m \exp(C_m L_m T)
\end{equation}
The evolution system properties (i) and (ii) from \cref{def-evolution-system} follow directly from the semigroup properties of the ${U_{k}}^{(i)}$ and uniform convergence of $U_k(t,s)$ in $s,t$ which allows us to exchange limits. Finally we have to show that the evolution semigroup is a solution to the Cauchy problem $\partial_t U(t,s) = G(t)U(t,s)$ (the $\partial_s$ version can be handled equivalently). Again we use uniform convergence and interchange time differentiation at $t \neq t_i \in P_k$ and the limit for the sequence $U_k(t,s)$.
\begin{equation}
\partial_t U(t,s) = \partial_t \lim_{k\rightarrow \infty} U_k(t,s) = \lim_{k\rightarrow \infty} \partial_t U_k(t,s) = \lim_{k\rightarrow \infty} G(\lfloor t \rfloor_k)U_k(t,s)
\end{equation}
On $D(A)$ we have $U_k(t,s) \rightarrow U(t,s) \in \mathcal{B}(D(A),D(A))$ as well as $G(\lfloor t \rfloor_k) \rightarrow G(t) \in \mathcal{B}(D(A),X)$ so we can establish the limits independently and get the desired evolution system for the Cauchy problem $\partial_t U(t,s) = G(t)U(t,s)$. If $t=t_i \in P_k$ the right and left derivatives will differ and yield $G(t_i)$ and $G(t_{i-1})$ respectively but in the limit $k \rightarrow \infty$ they are equal again because of the assumed continuity of the generator $G$ in time.
\end{proof}

\begin{note}
In \citet{schmid-griesemer-2016} it is assumed that $G(t)$ is the generator of a \textit{group} instead of a semigroup to get solutions to the equation involving the time derivative $\partial_t$ instead of just the right derivative $\partial_t^+$.
%(Maybe the last steps of the proof given here are missing such a difference.)
\end{note}

%\begin{note}
%A note on generalised solutions with $u_0 \in X \setminus D(A)$?
%\end{note}

%% not sure anymore if this is true, the W_k are shown to converge on the whole H by their definition and use of U_k convergence to an operator D(A) -> D(A)
%% the kato-1953 reference was moved to the lemma-weak-convergence
%\begin{note}
%In the proof of \citet[Th.~X.70]{reed-simon-2}, which resembles \cref{th-regularity} for only $m=1$, the ``similar proof'' (p.~288) showing convergence of $W_k(t,s)=A(t)U_k(t,s)A^{-1}(s)$ like for $U_k$ does not seem to work because the additional $A^{-1}$ would demand $\mathcal{B}(D(A^2),D(A))$ instead of just $\mathcal{B}(D(A),X)$ from the generator. This issue is solved here with the reflexive Banach space technique of \cref{lemma-weak-convergence} which is a weaker version of Lemma 5 in \citet{kato-1953} used on p.~222 top. It probably relates strongly to the statement: ``The closed unit ball in a normed space X is compact in the weak topology if and only if X is reflexive.'' (\href{https://en.wikipedia.org/wiki/Weak_topology\#Other_properties_2}{Wikipedia})
%\end{note}

\begin{note}
Establishing evolution systems between the different orders of graph-norm spaces $D(A^k)$ bears strong resemblance to the construction of so-called ``Sobolev towers'' in \citet[II.5.a]{engel-nagel}, although there the construction is only for time-independent generators.
\end{note}

\begin{note}
A similar proof strategy can be employed to show Fréchet differentiability of the solution to \eqref{eq-evolution-equation} in a Banach space including the time variable with respect to the perturbations $B$. See \citet[Th.~4.10]{penz-phd} for such a result and \citet{penz-2015} for a similar result using the completely different proof method of ``successive substitutions''.
\end{note}

\section{Skew-adjointness, generation theorems, and graph norm equivalents in Sobolev spaces}
\label{sec-graph-norm-equivalents}

The first lemma is standard and introduces the new default category for generators, i.e., skew-adjoint operators on a Hilbert space $\H$ with inner product $\langle \cdot,\cdot \rangle$. This notion has already been used in \cref{note-unitary-evolution-system} and means that $A$ is skew-adjoint if $\i A$ is self-adjoint or simply $A^* = -A$. The second result is the famous Stone theorem, see for example \citet[Th.~II.3.24]{engel-nagel}, which can be seen as a special case of the so-called generation theorems for the Banach space setting like Hille--Yosida \citet[Th.~II.3.5]{engel-nagel}. Generation theorems give conditions on the generators such that certain types of semigroups arise. This means the first lemma yielding closedness could be also replaced by the comment in the beginning of the proof of \cref{lemma-iterations-closed-2}, if the generator links to a strongly continuous semigroup. Finally we prove equivalence of the graph norm of the iterated Laplacian (or $\i \Delta$ that is a skew-adjoint operator) and the usual Sobolev norm.

\begin{lemma}\label{lemma-sa-closed}
A self-adjoint or skew-adjoint operator is always closed.
\end{lemma}

\begin{proof}
For $A$ self-adjoint and $z \in D(A)$ it holds with the same notation as in \cref{def-closed} that $\langle A x_n, z \rangle = \langle x_n, A z \rangle$ and therefore in the limit $n \rightarrow \infty$ we get $\langle y, z \rangle = \langle x, Az \rangle$. For this reason we have $x \in D(A)$ and $y = Ax$. The proof for skew-adjoint operators just introduces a minus sign that gets absorbed again when identifying $y=Ax$.
\end{proof}

\begin{theorem}[Stone]\label{th-stone}
Let $A$ be densely defined on $\H$. Then $A$ is the generator of a unitary (and thus contraction) group if and only if $A$ is skew-adjoint.
\end{theorem}

\begin{theorem}\label{th-sobolev-norm-laplace}
For general domains $\Omega \subseteq \R^n$ and $m \in \N$, considering the Hilbert space $W^{2m,2}(\Omega) \cap W^{m,2}_0(\Omega)$ the standard Sobolev norm is equivalent to the graph norm of $D(\Delta^m)$.
\begin{equation}
\|u\|_{2m,2} \sim \|u\|_{D(\Delta^m)}
\end{equation}
\end{theorem}

\begin{proof}
First observe that for arbitrary $u,v \in L^2(\Omega)$ it holds with the  Cauchy--Schwarz inequality and the inequality of arithmetic and geometric means that
\begin{equation}\label{eq-uv-estimate}
|\langle u,v \rangle| \leq \|u\|_2 \|v\|_2 \leq \onehalf (\|u\|_2^2 + \|v\|_2^2).
\end{equation}
The relation 
\begin{equation}
\sum_{l=0}^m \|\Delta^l u\|_2 \sim \left( \sum_{l=0}^m \|\Delta^l u\|_2^2 \right)^{1/2} \leq \|u\|_{2m,2}
\end{equation}
is fairly obvious. But can we also establish an estimate
\begin{equation}
\sum_{|\alpha| = k} \|D^\alpha u\|_2^2 \leq C \sum_{l=0}^m \|\Delta^l u\|_2^2
\end{equation}
in the other direction for all $0 < k \leq 2m$? If $k$ odd we use integration by parts ($u \in H_0^{m}$ is enough such that all boundary terms vanish) to get with \eqref{eq-uv-estimate}
\begin{equation}
\|D^\alpha u\|_2^2 = \langle D^\alpha u, D^\alpha u \rangle = |\langle D^{\alpha_1} u, D^{\alpha_2} u \rangle| \leq \onehalf (\|D^{\alpha_1} u\|_2^2 + \|D^{\alpha_2} u\|_2^2)
\end{equation}
where now $|\alpha_1|, |\alpha_2|$ even. For even $|\alpha|$ we proceed inductively and start with $|\alpha|=2$ and write out all partial derivatives separately. Integration by parts then yields
\begin{equation}\label{eq-sum-laplace-norm}
\begin{aligned}
\sum_{|\alpha| = 2} \|D^\alpha u\|_2^2 &= \sum_{|\alpha| = 2} \langle D^\alpha u, D^\alpha u \rangle \\
&= \sum_{i=1}^n \langle \partial_i^2 u, \partial_i^2 u \rangle + \sum_{i \neq j} \langle \partial_i \partial_j u, \partial_i \partial_j u \rangle \\
&= \sum_{i=1}^n \langle \partial_i^2 u, \partial_i^2 u \rangle + \sum_{i \neq j} \langle \partial_i^2 u, \partial_j^2 u \rangle \\
&= \sum_{i=1}^n \langle \partial_i^2 u, \partial_i^2 u + \sum_{j\neq i}\partial_j^2 u \rangle \\[0.5em]
&= \langle \Delta u,\Delta u \rangle = \|\Delta u\|_2^2.
\end{aligned}
\end{equation}
For $|\alpha|>2$ even we have to repeat the argument taking $u = \Delta v$ which means that
\begin{equation}
\|\Delta^2 v\|_2^2 = \sum_{|\alpha| = 2} \|D^\alpha \Delta v\|_2^2 = \sum_{|\alpha| = 2} \|\Delta D^\alpha v\|_2^2 = \sum_{|\alpha|, |\beta| = 2} \|D^{\alpha + \beta} v\|_2^2.
\end{equation}
Now $D^{\alpha + \beta}$ with all possible $|\alpha|=|\beta|=2$ includes all derivatives of order 4, some even multiple times. So we have the estimate
\begin{equation}
\sum_{|\alpha| = 4} \|D^{\alpha} v\|_2^2 \leq C \|\Delta^2 v\|_2^2
\end{equation}
that continues likewise to higher even $|\alpha|>4$.\footnote{This proof was inspired by answers in the following two StackExchange Mathematics threads:
\url{http://math.stackexchange.com/questions/101021} and
\url{http://math.stackexchange.com/questions/301404}.
}
\end{proof}

\begin{note}
The zero boundary condition $u \in W_0^{m,2}$ might also be replaced by a periodic domain where boundary terms vanish when integrating by parts. Note that in the theorem above the particular properties of the domain $\Omega$ are not of interest as it is usually the case considering $W^{m,p}_0$ Sobolev spaces because the respective (test) function can just be extended to all of $\R^n$ with zero, see \citet[4.12 III and 3.27]{adams}. A similar result on bounded domains including the graph norm of more general elliptic partial differential operators and the associated weak solutions of an inhomogeneous problem is called ``boundary regularity'' in \citet[6.3.2]{evans}. The even more general setting of elliptic partial differential operators of any order on compact manifolds is discussed in \citet[ch.~III, Th.~5.2 (iii)]{lawson-michelsohn}. In the wider literature similar results are known under the names ``G{\aa}rding inequality'' and ``fundamental elliptic estimate''.
\end{note}

%Add counterexamples from Maz'ya in $W^{2,2}$?

\section{Application to the Schrödinger equation}
\label{sec-se}

To treat the quantum mechanical case of particles in singular Coulombic potentials and other unbounded potentials we make use of the following lemma from Fourier analysis. Here the number of dimensions of the underlying space actually plays a crucial role and we are limited to dimension $n \leq 3$ for the one-particle configuration space in all further results because of the following lemma.

\begin{lemma}\label{lemma-ab-inequ}\citep[Th.~IX.28]{reed-simon-2}\\
Let $\varphi \in W^{2,2}(\R^n)$, $n\leq 3$. Then for all $\alpha >0$ there is a $\beta >0$ independent of $\varphi$ such that
\begin{equation}
\|\varphi\|_\infty \leq \alpha\|\Delta \varphi\|_2 + \beta\|\varphi\|_2.
\end{equation}
\end{lemma}

The next theorem is then a standard application of \cref{lemma-ab-inequ} together with the Kato--Rellich theorem to the case of the Schrödinger Hamiltonian with zero boundary conditions, see \citet[Th.~X.12]{reed-simon-2} and \citet[Th.~V.4.11]{kato-book}. The Kato--Rellich theorem states that if $A$ self-adjoint and $B$ symmetric then $A+B$ is also self-adjoint whenever $B$ is $A$-bounded with relative bound strictly smaller than $1$. The critical condition is thus that the potential turns out to be $\Delta$-bounded. The spatial domain $\Omega$ is always assumed to be a (open and connected) subset of $\R^n$, $n \leq 3$.

\begin{theorem}\label{th-kato}\citep[Th.~X.15]{reed-simon-2}\\
Given a real potential $v \in L^2(\R^n) + L^\infty(\R^n)$, $n \leq 3$, the Hamiltonian $-\Delta+v$ is self-adjoint on $W^{2,2}(\Omega) \cap W^{1,2}_0(\Omega)$.
\end{theorem}

\begin{definition}\label{def-kato-perturbations}
The space of \textbf{Kato perturbations} $L^2(\R^n) + L^\infty(\R^n)$ is equipped with the norm
\begin{equation}
\|v\|_{2+\infty} = \inf\{ \|v_1\|_2 + \|v_2\|_\infty \mid v_1 \in L^2(\R^n), v_2 \in L^\infty(\R^n), v=v_1+v_2\}.
\end{equation}
\end{definition}

The following notation for the extension of potentials to multi-particle systems with $N$ particles is borrowed from \citet{lammert}. Note that in the published version of this work that we cite along the preprint the respective notation has vanished again.

\begin{definition}
For a one-point function $v : \Omega \rightarrow \R$ we define
\begin{equation}
\Gamma v : \Omega^N \rightarrow \R, (x_1,\ldots,x_N) \mapsto \sum_{i=1}^N v(x_i)
\end{equation}
and similarly for a two-point function $w : \Omega \times \Omega \rightarrow \R$
\begin{equation}
\Gamma w : \Omega^N \rightarrow \R, (x_1,\ldots,x_N) \mapsto \frac{1}{2} \sum_{\substack{i,j=1 \\ i\neq j}}^N w(x_i,x_j).
\end{equation}
\end{definition}

\begin{lemma}\label{lemma-kato-perturbations-estimate}
Given the potentials $v,\vint \in L^2(\R^n) + L^\infty(\R^n)$, $n \leq 3$, and the interaction potential $w(x_1,x_2) = \vint(x_1-x_2)$ the multiplication operators $\Gamma v$ and $\Gamma w$ are both $\Delta$-bounded with relative bound 0. There is further a constant $\beta > 0$ such that the following estimates hold for all $\varphi \in W^{2,2}(\Omega^N)$.
\begin{align}
\| (\Gamma v)\varphi \|_2 &\leq N\beta \|v\|_{2+\infty} \|\varphi\|_{2,2} \\
\| (\Gamma w)\varphi \|_2 &\leq \frac{N(N-1)}{2}\beta \|\vint\|_{2+\infty} \|\varphi\|_{2,2}
\end{align}
\end{lemma}

\begin{proof}
We adopt the following notation for the norm of the Hilbert space $L^2(\Omega)$ where we assume all coordinates $x_{j\neq i}$ fixed and analogously if only one coordinate $x_i$ is fixed.
\begin{equation}
\|\varphi\|_2^{(i)} = \left( \int_\Omega |\varphi|^2 \d x_i \right)^{1/2}, \quad
\|\varphi\|_2^{(j\neq i)} = \left( \int_\Omega |\varphi|^2 \d x_{j\neq i} \right)^{1/2}
\end{equation}
Note that it holds $\|\varphi\|_2 = \|\|\varphi\|_2^{(i)}\|_2^{(j\neq i)}$ so that we have
\begin{equation}\label{eq-kato-perturbation-two-norms}
\| (\Gamma v)\varphi \|_2 \leq \sum_{i=1}^N \|v(x_i)\varphi\|_2 = \sum_{i=1}^N \|\|v(x_i)\varphi\|_2^{(i)}\|_2^{(j\neq i)}.
\end{equation}
Now the inner norm is estimated with the decomposition $v=v_1 + v_2$, $v_1 \in L^2$, $v_2 \in L^\infty$ as $\|v(x_i)\varphi\|_2^{(i)} \leq \|v_1\|_2 \|\varphi\|_\infty^{(i)} + \| v_2 \|_\infty \|\varphi\|_2^{(i)}$ where we use the obvious notation of $\|\varphi\|_\infty^{(i)}$ as the essential supremum of $\varphi$ over all $x_i \in \Omega$. Note that $\|v_1\|_2$ and $\| v_2 \|_\infty$ are just numbers with no free variables left. It is now time to invoke \cref{lemma-ab-inequ} and have for arbitrarily small $\alpha_i>0$
\begin{equation}\label{eq-use-lemma-ab-inequ}
\|\varphi\|_\infty^{(i)} \leq \alpha_i \|\Delta_i \varphi\|_2^{(i)} +  \beta_i \|\varphi\|_2^{(i)}.
\end{equation}
Combination of these estimates gives
\begin{equation}
\| (\Gamma v)\varphi \|_2 \leq \sum_{i=1}^N \left( \alpha_i \|v_1\|_2 \|\Delta_i \varphi\|_2 + (\beta_i \|v_1\|_2 + \| v_2 \|_\infty) \|\varphi\|_2 \right).
\end{equation}
A final trick is needed to have the full Laplacian $\Delta$ instead of $\Delta_i$ only involving $x_i$. For this we observe that by moving to the Fourier domain with coordinates $k_i \in \R^n$
\begin{equation}
\|\Delta_i \varphi\|_2 = \|k_i^2 \hat\varphi \|_2 \leq \Big\| \sum_{j=1}^N k_j^2 \hat\varphi \Big\|_2 = \|\Delta \varphi\|_2.
\end{equation}
Now define $\alpha = \max_i \alpha_i$ (but still arbitrarily small) and $\beta = \max \{ \beta_1, \ldots, \beta_N, \alpha, 1 \}$ and we get
\begin{equation}
\| (\Gamma v)\varphi \|_2 \leq N \alpha \|v_1\|_2 \|\Delta \varphi\|_2 + N(\beta \|v_1\|_2 + \| v_2 \|_\infty) \|\varphi\|_2.
\end{equation}
This means $\Gamma v$ is $\Delta$-bounded with relative bound 0. If we further introduce $\|v\|_{2+\infty}$ and choose $v_1$ and $v_2$ accordingly then with $\beta$ defined as above we can take it out as an upper estimate. Together with the equivalence of the graph norm of $D(\Delta)$ and the Sobolev norm $\|\cdot\|_{2,2}$ from \cref{th-sobolev-norm-laplace} we arrive at the desired
\begin{equation}
\| (\Gamma v)\varphi \|_2 \leq N\beta \|v\|_{2+\infty} \|\varphi\|_{2,2}.
\end{equation}
The proof for the two-point potential that is defined as an interaction potential involving $\vint(x_i-x_j)$ is analogous but one first has to rotate the whole $\Omega^N \subseteq \R^{nN}$ so that $x_i-x_j$ matches the $x_1$ coordinate. This is possible invariantly because the $L^2$-norm is rotational invariant. The rest of the proof stays the same, we only consider $N(N-1)/2$ components in the sum instead of only $N$.
\end{proof}

\begin{note}
The \cref{lemma-kato-perturbations-estimate} above allows for an extension of \cref{th-kato} to multi-particle systems with Hamiltonian $H = -\Delta + \Gamma w + \Gamma v$ if the involved potentials are of type $v,\vint \in L^2(\R^n) + L^\infty(\R^n)$ with $w(x_i,x_j) = \vint(x_i-x_j)$. The proof structure of \cref{lemma-kato-perturbations-estimate} was inspired by a theorem with this assertion given in \citet[Th.~X.16]{reed-simon-2}.
\end{note}

\begin{definition}\label{def-sobolev-kato-perturbations}
We extend \cref{def-kato-perturbations} (Kato perturbations) to \textbf{Sobolev--Kato perturbations}, defined as the space of potentials
\begin{equation}
W^{m,2+\infty}(\R^n) = \{ v \mid D^\alpha v \in L^2(\R^n)+L^\infty(\R^n), |\alpha|\leq m \}
\end{equation}
with norm
\begin{equation}
\|v\|_{m,2+\infty} = \sum_{|\alpha| \leq m} \| D^\alpha v \|_{2+\infty}.
\end{equation}
\end{definition}

\begin{lemma}\label{lemma-sobolev-kato-perturbations-estimate}
Given the potentials $v,\vint \in W^{2m,2+\infty}(\R^n)$, $n \leq 3$, and $w(x_1,x_2) = \vint(x_1-x_2)$ the multiplication operators $\Gamma v$ and $\Gamma w$ are both $\Delta$-bounded of order $m+1$ with relative bound 0. There is further a constant $\beta > 0$ such that the following estimates hold for all $\varphi \in W^{2(m+1),2}(\Omega^N)$.
\begin{align}
\| (\Gamma v)\varphi \|_{2m,2} &\leq N\beta \|v\|_{2m,2+\infty} \|\varphi\|_{2(m+1),2} \\
\| (\Gamma w)\varphi \|_{2m,2} &\leq \frac{N(N-1)}{2}\beta \|\vint\|_{2m,2+\infty} \|\varphi\|_{2(m+1),2}
\end{align}
\end{lemma}

\begin{proof}
\Cref{lemma-kato-perturbations-estimate} already shows the case $m=0$ and proceeding from that we give the proof for arbitrary orders $m$. We start by writing out the involved Sobolev space norm explicitly, then employ the general Leibniz rule for multivariable calculus.
\begin{equation}
\begin{aligned}
\|(\Gamma v) \varphi\|_{2m,2} \sim \smashoperator{\sum_{|\alpha| \leq 2m}} \left\| D^\alpha ((\Gamma v) \varphi) \right\|_2 &= \sum_{|\alpha| \leq 2m} \Big\| \sum_{\nu \leq \alpha} {\alpha \choose \nu} \left(D^\nu (\Gamma v)\right)  D^{\alpha-\nu} \varphi \Big\|_2 \\
&\leq {2m \choose m}  \sum_{|\alpha| \leq 2m} \sum_{\nu \leq \alpha} \left\| \left(D^\nu (\Gamma v)\right) D^{\alpha-\nu} \varphi \right\|_2
\end{aligned}
\end{equation}
The multi-index binomial coefficient is estimated by its largest possible value. Next we use the property of $\Gamma v$ that makes the potential the sum of one-coordinate potentials. Thus instead of the full $D^\nu$ only $D^{\nu_i}$ acts on the individual terms of the sum in $\Gamma v$. Note that these $\nu_i$ from $\nu = (\nu_1,\ldots,\nu_N)$ are still $n$-tuples. In any case we have
\begin{equation}
D^\nu (\Gamma v) = \sum_{i=1}^N D^{\nu_i} v(x_i)
\end{equation}
and thus
\begin{equation}
\begin{aligned}
\smashoperator{\sum_{|\alpha| \leq 2m}} \left\| D^\alpha ((\Gamma v) \varphi) \right\|_2 &\leq {2m \choose m}  \sum_{|\alpha| \leq 2m} \sum_{\nu \leq \alpha} \sum_{i=1}^N \left\| \left( D^{\nu_i} v(x_i) \right) D^{\alpha-\nu} \varphi \right\|_2 \\
&=  {2m \choose m}  \sum_{|\alpha| \leq 2m} \sum_{\nu \leq \alpha} \sum_{i=1}^N \left\| \left\| \left( D^{\nu_i} v(x_i) \right) D^{\alpha-\nu} \varphi \right\|_2^{(i)} \right\|_2^{(j\neq i)}
\end{aligned}
\end{equation}
like in \eqref{eq-kato-perturbation-two-norms}. The proof then proceeds exactly like in \cref{lemma-kato-perturbations-estimate} since we have $D^{\nu_i} v(x_i) \in L^2+L^\infty$ due to the assumption $v \in W^{m,2+\infty}$. In total we get the estimate
\begin{equation}
\smashoperator{\sum_{|\alpha| \leq 2m}} \left\| D^\alpha ((\Gamma v) \varphi) \right\|_2 \leq N \beta' {2m \choose m} \|v\|_{2m,2+\infty} \|\varphi\|_{2(m+1),2}
\end{equation}
where the sums over the multi-indices get combined and estimated by the higher Sobolev norms. The order of the Sobolev norm of $\varphi$ has increased by 2 because we had to rely on \cref{lemma-ab-inequ} again. The constant $\beta'$ that is defined similar as in \cref{lemma-kato-perturbations-estimate} before gets combined together with the binomial coefficient to form a constant $\beta$ and we arrive at the desired result.
If we keep the arbitrarily small $\alpha_i$ that are introduced analogously to \eqref{eq-use-lemma-ab-inequ}, then this also yields the desired $\Delta$-boundedness of order $m+1$ with relative bound 0. In both cases the equivalence of the graph norm of $D(\Delta^m)$ and the Sobolev norm $\|\cdot\|_{2m,m}$ from \cref{th-sobolev-norm-laplace} gets applied.\\
The way for a two-point potential is the same as before with the only difference that we have to observe
\begin{equation}
D^\nu (\Gamma w) = \frac{1}{2} \sum_{\substack{i,j=1 \\ i\neq j}} D^{\nu_i}D^{\nu_j} \vint(x_i-x_j)
\end{equation}
before rotating $x_i-x_j$ again so that it matches the $x_1$ coordinate in the individual contributions of the norm.
\end{proof}

\begin{theorem}\label{th-se}
The $N$-particle Schrödinger equation
\begin{equation}
\i\partial_t \psi(t) = H(t)\psi(t) = (-\Delta + \Gamma w + \Gamma v(t)) \psi(t)
\end{equation}
on the Hilbert space $L^2(\Omega^N)$, $\Omega \subseteq \R^n$ open and connected, $n \leq 3$, with
\begin{equation}
v \in \mathrm{Lip}([0,T], W^{2(m-1),2+\infty}(\R^n))
\end{equation}
and $w(x_i,x_j) = \vint(x_i-x_j)$, $\vint \in W^{2(m-1),2+\infty}(\R^n)$ has a well-defined unitary evolution system that is bounded as a mapping
\begin{equation}
U(t,s) : W^{2m,2}(\Omega^N) \cap W^{m,2}_0(\Omega^N) \rightarrow W^{2m,2}(\Omega^N) \cap W^{m,2}_0(\Omega^N), \quad t,s \in [0,T].
\end{equation}
This effectively establishes Sobolev regularity of solutions up to order $W^{2m,2}$.
\end{theorem}

\begin{proof}
We rewrite the Schrödinger equation as
\begin{equation}
\partial_t \psi(t) = -\i H(t)\psi(t) = \i(\Delta - \Gamma w - \Gamma v(t)) \psi(t)
\end{equation}
and take $A=\i \Delta$, $G(t) = -\i H(t)$. Both operators are generators of contraction semigroups because of \cref{th-stone} (Stone) in conjunction with \cref{lemma-kato-perturbations-estimate} and the Kato--Rellich theorem. \Cref{lemma-sobolev-kato-perturbations-estimate} tells us that $\Gamma v(t)$ and $\Gamma w$ are bounded operators $W^{2m,2}(\Omega^N) \rightarrow W^{2(m-1),2}(\Omega^N)$ as well as $\Delta$-bounded of order $m$ with relative bound 0. Thus all the requirements on the non-autonomous perturbation from \cref{lemma-K-bounded} are fulfilled and \cref{th-regularity} becomes applicable which establishes the desired regularity result.
\end{proof}

\begin{note}
A final note shall make the setting even more ``physical'' and turns attention towards the standard example for external potentials and interactions, the Coulomb potential on $\Omega = \R^3$. If we take $v(x) = -|x|^{-1}$ (attractive) or $\vint(x) = |x|^{-1}$ (repulsive) those potentials lie in the class of Kato perturbations $L^2(\R^3) + L^\infty(\R^3)$. Even the more singular choice of $|x|^{-3/2+\varepsilon}$ for arbitrarily small $\varepsilon$ is permitted. Further the potentials can be time-dependent under the constraint of the introduced Lipschitz condition. But those potentials already drop out of the next higher regularity class $W^{1,2+\infty}(\R^3)$ thus \cref{th-se} only guarantees Sobolev regularity up to $W^{2,2}$ for solutions to the Schrödinger equation with Coulomb potentials. 
\end{note}

\section*{Acknowledgements}

With special thanks to Jonas Lampart for pointing out the works of Schmid--Griesemer, as well as to Michael Ruggenthaler and Eric Stachura for helpful discussions. An anonymous referee spotted a critical flaw in the proof of \cref{lemma-K-bounded}, that in the course led to the new \cref{def-rel-boundedness} and {\cref{lemma-AG-equiv}}, and helped to improve the manuscript substantially. This work was supported by the Erwin Schrödinger Fellowship J 4107-N27 of the FWF (Austrian Science Fund).

\end{document}